\documentclass[conference]{IEEEtran}
\IEEEoverridecommandlockouts
\usepackage{cite}
\usepackage{amsmath,amssymb}
\usepackage{graphicx}
\usepackage{booktabs}
\usepackage{algorithm}
\usepackage{algpseudocode}
\usepackage{glossaries}
\usepackage{balance}

\newtheorem{theorem}{Theorem}
\newtheorem{lemma}{Lemma}
\newtheorem{proposition}{Proposition}
\DeclareMathOperator{\Tr}{Tr}

\newacronym{5g}{5G}{fifth-generation wireless technology}
\newacronym{6g}{6G}{sixth-generation wireless technology}

\newacronym{oma}{OMA}{orthogonal multiple access}
\newacronym{noma}{NOMA}{non-orthogonal multiple access}
\newacronym{pd-noma}{PD-NOMA}{power-domain non-orthogonal multiple access}
\newacronym{rsma}{RSMA}{rate-splitting multiple access}
\newacronym{fdma}{FDMA}{frequency division multiple access}
\newacronym{tdma}{TDMA}{time division multiple access}
\newacronym{cdma}{CDMA}{code division multiple access}
\newacronym{sdma}{SDMA}{space division multiple access}
\newacronym{ofdma}{OFDMA}{orthogonal frequency division multiple access}
\newacronym{sc}{SC}{superposition coding}
\newacronym{sic}{SIC}{successive interference cancellation}

\newacronym{ris}{RIS}{reconfigurable intelligent surface}
\newacronym{irs}{IRS}{intelligent reflecting surface}
\newacronym{d-ris}{D-RIS}{diagonal reconfigurable intelligent surface}
\newacronym{bd-ris}{BD-RIS}{beyond-diagonal reconfigurable intelligent surface}
\newacronym{star-ris}{STAR-RIS}{simultaneously transmitting and reflecting
  reconfigurable intelligent surface}

\newacronym{backcom}{BackCom}{backscatter communication}
\newacronym{ambc}{AmBC}{ambient backscatter communication}
\newacronym{abcs}{ABCS}{ambient backscatter communication systems}
\newacronym{bd}{BD}{backscatter device}
\newacronym{br}{BR}{backscatter receiver}
\newacronym{sr}{SR}{symbiotic radio}

\newacronym{bs}{BS}{base station}
\newacronym{pt}{PT}{primary transmitter}
\newacronym{pr}{PR}{primary receiver}
\newacronym{cr}{CR}{cognitive radio}
\newacronym{iot}{IoT}{Internet of Things}

\newacronym{siso}{SISO}{single-input single-output}
\newacronym{miso}{MISO}{multiple-input single-output}
\newacronym{mimo}{MIMO}{multiple-input multiple-output}
\newacronym{tas}{TAS}{transmit antenna selection}
\newacronym{bf}{BF}{beamforming}

\newacronym{csi}{CSI}{channel state information}
\newacronym{snr}{SNR}{signal-to-noise ratio}
\newacronym{sinr}{SINR}{signal-to-interference-plus-noise ratio}
\newacronym{awgn}{AWGN}{additive white Gaussian noise}
\newacronym{bpsk}{BPSK}{binary phase shift keying}

\newacronym{se}{SE}{spectral efficiency}
\newacronym{ec}{EC}{ergodic capacity}
\newacronym{ber}{BER}{bit error rate}
\newacronym{op}{OP}{outage probability}
\newacronym{cop}{COP}{coexistence outage probability}
\newacronym{wsr}{WSR}{weighted sum rate}
\newacronym{qos}{QoS}{quality of service}
\newacronym{ep}{EP}{enhancing primary signal}

\newacronym{ao}{AO}{alternating optimisation}
\newacronym{sca}{SCA}{successive convex approximation}
\newacronym{sdr}{SDR}{semidefinite relaxation}
\newacronym{sdp}{SDP}{semidefinite program}
\newacronym{srocr}{SROCR}{sequential rank-one constraint relaxation}
\newacronym{pac}{PAC}{power allocation coefficient}
\newacronym{np}{NP}{nondeterministic polynomial time}

\begin{document}

\title{Beamforming Design for RIS-Assisted Backscatter NOMA in
Underlay Cognitive Radio}
\author{%
\IEEEauthorblockN{Kgomotjo Seipobi\IEEEauthorrefmark{1},
Fambirai Takawira\IEEEauthorrefmark{2} and
Chabalala Chabalala\IEEEauthorrefmark{3}}
\IEEEauthorblockA{\textit{School of Electrical and Information Engineering}\\
\textit{University of the Witwatersrand}, Johannesburg, South Africa\\
\IEEEauthorrefmark{1}2919396@students.wits.ac.za,
\IEEEauthorrefmark{2}fambirai.takawira@wits.ac.za,
\IEEEauthorrefmark{3}chabalala.chabalala@wits.ac.za}}

\maketitle

\begin{abstract}
This paper investigates a reconfigurable intelligent surface (RIS)-assisted
backscatter non-orthogonal multiple access (NOMA) network operating as the
secondary system of an underlay cognitive radio network. A multi-antenna base
station serves several user clusters through the RIS, each containing a
backscatter device (BD). We maximise the weighted sum rate of the secondary network by
jointly optimising the transmit beamformers and the RIS reflection matrix,
subject to the power budget, the interference-temperature limit and minimum
rate requirements. The non-convex problem is solved by alternating
optimisation, using semidefinite relaxation and sequential rank-one
constraint relaxation. Interference-normalised
user ordering makes the successive interference cancellation constraints
redundant. A closed-form equal-rate power split is derived; it
agrees closely with simulation and costs about 4.5\,\% of the sum rate.
Supporting the BDs reduces the NOMA sum rate by 8.7\,\% relative to a no-BD
design, while a larger BD reflection coefficient raises the backscatter rate
with little NOMA loss.
\end{abstract}


\glsresetall

\section{Introduction}
\label{sec:intro}

The proliferation of the Internet of Things has intensified the demand
for spectrum- and energy-efficient wireless access~%
\cite{zhang_spectrum_2019}. Several complementary technologies have
emerged to address this demand. \Glspl{ris} can enhance connectivity
when direct links are obstructed by reconfiguring the propagation
environment through low-cost passive elements without requiring
additional radio-frequency chains~\cite{di_renzo_smart_2020}.
\Gls{noma} enables multiple users to share the same resource block
through power-domain multiplexing and \gls{sic}, thereby improving
spectral utilisation and user connectivity~%
\cite{Aldababsa_tutorial_noma,Liu_noma}. Backscatter communication
allows a battery-free device to transmit by modulating an incident
waveform~%
\cite{jiang_backscatter_2023,basharat_reconfigurable_2022,
asif_energy-efficient_2023}; when this waveform is supplied by a host
network, the resulting interaction forms a symbiotic radio system~%
\cite{liang_symbiotic_2020}. Underlay \gls{cr} further permits a
secondary network to reuse licensed spectrum, provided that the
interference caused at the primary receiver remains below a prescribed
threshold~\cite{Zhou_Artificial_cr}.

Joint active and passive beamforming is well established for
conventional \gls{ris}-assisted \gls{noma}~%
\cite{Mu_Exploiting,zuo_joint_2023,gao_joint_2023}, but these designs contain
neither \glspl{bd} nor a primary network.
\Gls{ris}-assisted \gls{noma}-enabled underlay ambient backscatter was
investigated in~\cite{thenua_sum_2024}, where multiple \glspl{bd} use the
primary waveform to communicate with a dedicated backscatter receiver.
However, the secondary network is restricted to backscatter communication: it
contains no secondary \gls{bs} or conventional \gls{noma}-user service.
Accordingly, only the aggregate \gls{bd} rate is maximised, without active
transmit beamforming.
Cognitive backscatter \gls{noma} designs employing active \glspl{ris} as
secondary transmitters~\cite{han_aris_cr_noma} likewise differ from the
physically separate-\gls{bd} symbiotic setting considered here.
A symbiotic architecture integrating physically separate \glspl{bd}
with a \gls{bs}-based \gls{noma} service was considered in~%
\cite{Asif_noma_zeRIS}. Nevertheless, that work assumes a standalone network
without a licensed primary system, so no interference-temperature constraint
arises. Moreover, fairness is
enforced there only through minimum \gls{sinr} constraints, so the
equal-rate intra-cluster power split is not characterised. Similarly,~\cite{sun_weighted_2024}
maximises the weighted sum rate of an
RIS-backscatter-aided NOMA network, but the RIS itself serves as the
information-bearing backscatter device. The same architecture is carried into an uplink underlay CR-NOMA network in~%
\cite{sun_sib_ris_2026}, which jointly optimises the primary-transmitter
beamformer, the reflection of a self-sustainable information-bearing RIS and
its energy-harvesting power split; there, too, the RIS carries the
information by modulating its own reflection. In
contrast, the present work considers a downlink in which a passive RIS assists
physically separate cluster-associated BDs within an underlay cognitive-radio network,
where the secondary \gls{bs} beamformers and the RIS are designed jointly and
must additionally satisfy the interference-temperature constraint protecting
the primary receiver. To the best of our knowledge, extending \gls{ris}-assisted \gls{noma}
beamforming to separate \glspl{bd} in an underlay \gls{cr} network has not been
reported. The equal-rate operating point is characterised
in~\cite{Agarwal_ndris_fairness}, but for a single-antenna access point
serving two-user pairs in orthogonal time slots without \glspl{bd}; under
multi-cluster \gls{miso} transmission it also depends on the beamformers and
the \gls{bd} interference, and neither work quantifies its throughput cost.

Motivated by these limitations, we consider a \gls{ris}-assisted
symbiotic \gls{miso}-\gls{noma} network operating as the secondary
system of an underlay \gls{cr} architecture. The conventional
\gls{noma} users and the physically separate \glspl{bd} convey
independent information over the same secondary waveform. Accordingly,
the \gls{bs} active beamformers and the \gls{ris} passive beamforming
coefficients are jointly optimised to maximise the weighted sum of the
\gls{noma}-user and \gls{bd} rates, subject to their respective
\gls{qos} requirements and the interference-temperature constraint
protecting the primary receiver. The resulting weighted-sum-rate
maximisation problem is highly non-convex because of the coupling between
the active and passive beamforming variables and the \gls{ris}
unit-modulus constraints. The main contributions are summarised as follows.
\begin{itemize}
\item We develop an \gls{ris}-assisted symbiotic \gls{miso}-\gls{noma} model
for an underlay \gls{cr} network in which cluster-associated \glspl{bd}
modulate their information onto the incident secondary waveform, so that a
single \gls{ris} configuration jointly shapes the \gls{noma} and the
backscatter links. Each \gls{bd} symbol is decoded last within its cluster,
at the user that has cancelled the cluster's \gls{noma} symbols and can reuse
them as a data-aided reference.

\item We formulate the joint active and passive beamforming problem for
weighted-sum-rate maximisation under the interference-temperature limit,
the transmit power budget, and the minimum rates of both the \gls{noma} users and
the \glspl{bd}. We show that the redundancy of the \gls{sic} constraints under
an interference-normalised ordering extends to the present setting, in which
the interference floor also carries the backscatter and cross-network terms,
so that those constraints can be removed from the problem without loss.

\item We propose an \gls{ao} algorithm that updates the \gls{bs} beamformers
by \gls{sca} and \gls{sdr}, whose rank-one tightness follows the argument
of~\cite{xu_resource_allocation_irs}, and the \gls{ris} phase shifts by
\gls{srocr}, which handles the unit-modulus constraint. Numerical results establish convergence within about $12$ iterations,
validate the derived closed-form equal-rate power split against the simulated
rate crossover, and show that operating at the equal-rate split costs about
$4.5\,\%$ of the sum rate. They further show that supporting the \glspl{bd}
reduces the \gls{noma} sum rate by $8.7\,\%$ relative to a no-\gls{bd} design, while increasing the
backscatter reflection coefficient enhances the backscatter rate with little
further \gls{noma}-rate loss.
\end{itemize}

\section{System Model and Problem Formulation}
\label{sec:system}

\begin{figure}[t]
    \centering
    \includegraphics[width=\columnwidth]{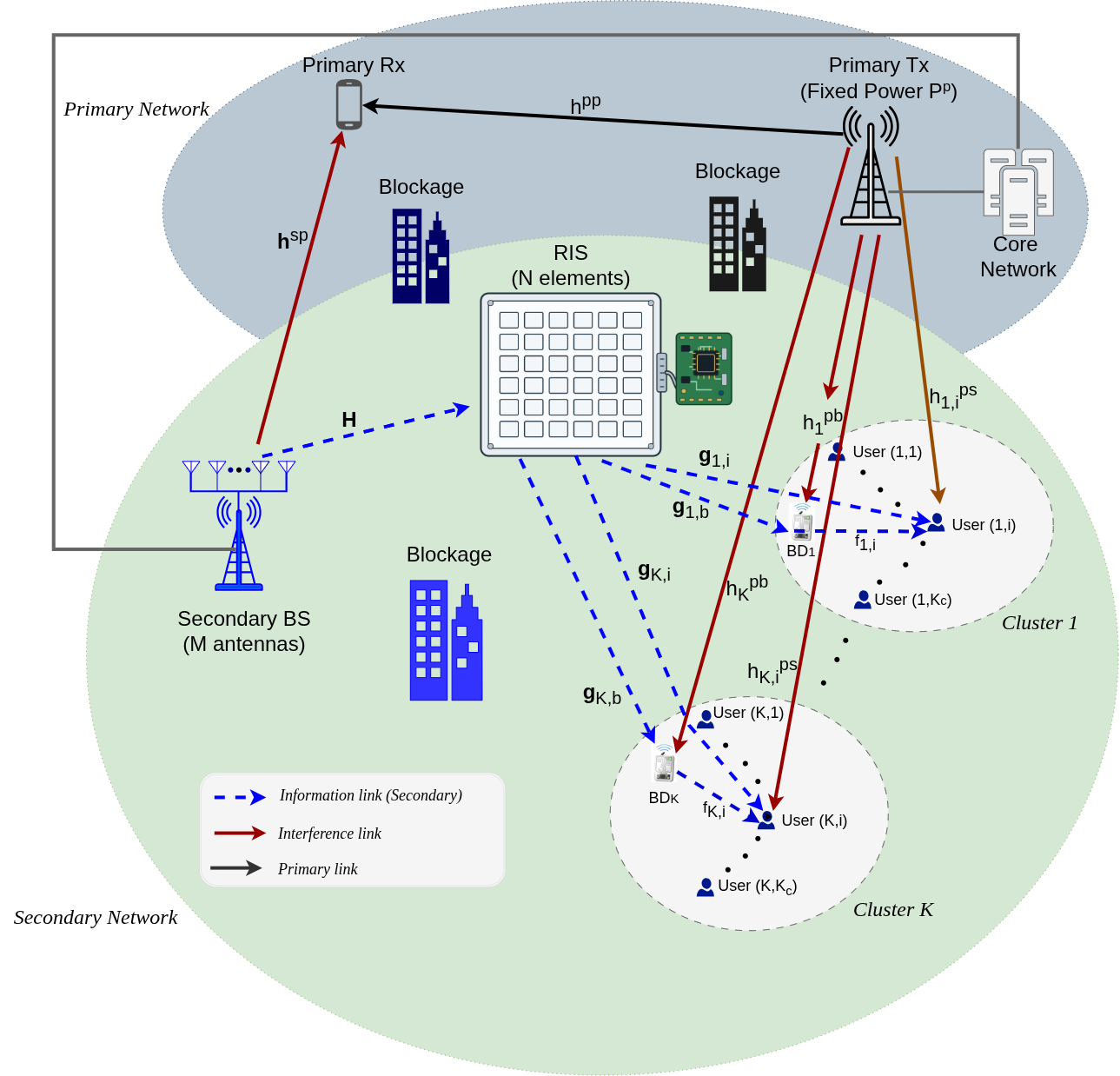}
    \caption{RIS-assisted symbiotic-backscatter MISO-NOMA network operating as the secondary of an underlay cognitive-radio system.}
    \label{fig:system_model}
\end{figure}

As shown in Fig.~\ref{fig:system_model}, we consider the downlink of an
\gls{ris}-assisted backscatter \gls{noma} network operating as the secondary of
an underlay \gls{cr} system. The primary network consists of a single-antenna
\gls{pt} transmitting at fixed power $P^{\mathrm{p}}$ and a single-antenna
\gls{pr}, whose operation is independent of the secondary design. The secondary
network comprises an $M$-antenna \gls{bs}, a passive \gls{ris} with $N$
reflecting elements, and $L=KK_c$ single-antenna users partitioned into $K$
clusters of $K_c$ users each, indexed by $k\in\mathcal{K}\triangleq\{1,\dots,K\}$
and $i\in\mathcal{K}_c\triangleq\{1,\dots,K_c\}$. Each cluster contains one
single-antenna \gls{bd} and is served by the beamforming vector
$\mathbf{w}_k\in\mathbb{C}^{M\times1}$. The
direct \gls{bs}--user and \gls{bs}--\gls{bd} links are blocked, so the
secondary transmission reaches its receivers only through the \gls{ris}, whose
configuration jointly shapes the \gls{noma} and backscatter channels.
The deployment is configured such that the \gls{ris} is shadowed from
the primary network. Consequently, no \gls{ris}-assisted cross-network
paths are considered: primary interference reaches the secondary receivers
over the direct \gls{pt}--secondary links, while secondary leakage reaches
the \gls{pr} over the direct \gls{bs}--\gls{pr} link. Within cluster $k$,
\gls{bd}$_k$ is illuminated by both the \gls{ris}-reflected secondary
waveform and the directly received primary signal, and modulates its own
information symbol $c_k$ onto the resulting incident field. The fraction
of incident power it re-radiates is set by the \gls{bd} reflection coefficient
$\eta_k\in[0,1]$, which governs the strength of the backscatter link and, as
shown below, the interference it imposes on the \gls{noma} users of the
cluster. Perfect \gls{csi} is assumed at the \gls{bs}.\footnote{The direct
cross-network \gls{csi} is obtained through the core network coordinating the
two tiers~\cite{Hui_robust_bf_cr}, and the cascaded \gls{ris}--\gls{bd} and
\gls{bd}--user coefficients using the backscatter channel-estimation
procedure of~\cite{liu_ambient_2013}.}

Denote by $\mathbf{H}\in\mathbb{C}^{N\times M}$,
$\mathbf{g}_{k,i}\in\mathbb{C}^{N\times1}$,
$\mathbf{g}_{k,b}\in\mathbb{C}^{N\times1}$, and
$f_{k,i}\in\mathbb{C}$ the \gls{bs}--\gls{ris},
\gls{ris}--user $(k,i)$, \gls{ris}--\gls{bd}$_k$, and
\gls{bd}$_k$--user $(k,i)$ channels, respectively. The direct
\gls{pt}--user, \gls{pt}--\gls{bd}$_k$, \gls{bs}--\gls{pr}, and
\gls{pt}--\gls{pr} channels are denoted by
$h^{ps}_{k,i}\in\mathbb{C}$, $h^{pb}_k\in\mathbb{C}$,
$\mathbf{h}^{sp}\in\mathbb{C}^{M\times1}$, and
$h^{pp}\in\mathbb{C}$, respectively. The \gls{ris}-related channels are
modelled as Rician fading, while the remaining channels follow Rayleigh
fading.

With
$\boldsymbol{\Theta}
=\mathrm{diag}(e^{j\theta_1},\dots,e^{j\theta_N})$,
where $\theta_n\in[0,2\pi)$, the effective \gls{bs}--user and cascaded
\gls{bs}--\gls{bd}$_k$--user channels for user $(k,i)$ are
\begin{equation}
\mathbf{h}_{k,i}^{H}
\triangleq
\mathbf{g}_{k,i}^{H}\boldsymbol{\Theta}\mathbf{H},
\qquad
\mathbf{h}_{k,i}^{c,H}
\triangleq
f_{k,i}\mathbf{g}_{k,b}^{H}\boldsymbol{\Theta}\mathbf{H}.
\label{eq:cascaded}
\end{equation}
Thus, $\boldsymbol{\Theta}$ affects the secondary links only through
\eqref{eq:cascaded}, while the direct cross-network interference links are
independent of the \gls{ris} configuration.

Employing superposition coding, the \gls{bs} transmits
$\mathbf{s}=\sum_{k=1}^{K}\mathbf{w}_k\sum_{i=1}^{K_c}
\sqrt{\alpha_{k,i}}\,s_{k,i}$, where $\mathbb{E}\{|s_{k,i}|^2\}=1$ and
the power allocation coefficients satisfy
$\sum_{i=1}^{K_c}\alpha_{k,i}=1$, while the \gls{pt} transmits
$x_p=\sqrt{P^{\mathrm{p}}}s_p$ with $\mathbb{E}\{|s_p|^2\}=1$. As a
passive scatterer, \gls{bd}$_k$ modulates its information symbol $c_k$,
with $\mathbb{E}\{|c_k|^2\}=1$, onto the field impinging on it,
comprising the \gls{ris}-reflected secondary signal and the directly
received primary signal. Its transmitted signal is
\begin{equation}
x^{b}_{k}=\sqrt{\eta_k}\,c_k\big(
\mathbf{g}_{k,b}^{H}\boldsymbol{\Theta}\mathbf{H}\,\mathbf{s}
+\sqrt{P^{\mathrm{p}}}\,h^{pb}_k\,s_p\big).
\label{eq:tx_bd}
\end{equation} The received
signal at user $(k,i)$ is then
\begin{equation}
\begin{aligned}
y_{k,i}=\;
&\underbrace{\sqrt{\alpha_{k,i}}\,\mathbf{h}_{k,i}^{H}\mathbf{w}_k\,s_{k,i}}_{\text{desired}}
+\underbrace{\mathbf{h}_{k,i}^{H}\mathbf{w}_k\!\!\sum_{m\neq i}^{K_c}\!\!
\sqrt{\alpha_{k,m}}\,s_{k,m}}_{\text{intra-cluster}}\\
&+\underbrace{\mathbf{h}_{k,i}^{H}\sum_{j\neq k}\mathbf{w}_j\sum_{m=1}^{K_c}
\sqrt{\alpha_{j,m}}\,s_{j,m}}_{\text{inter-cluster}}
+\underbrace{\sqrt{\eta_k}\,c_k\,\mathbf{h}_{k,i}^{c,H}\,\mathbf{s}}_{\text{backscatter}}\\
&+\underbrace{\sqrt{\eta_k P^{\mathrm{p}}}\,f_{k,i}\,h^{pb}_k\,c_k\,s_p
+\sqrt{P^{\mathrm{p}}}\,h^{ps}_{k,i}\,s_p}_{\text{backscattered and direct primary interference}}
+\;n_{k,i},
\end{aligned}
\label{eq:rx_sec}
\end{equation}
where $n_{k,i}\sim\mathcal{CN}(0,\sigma^2)$. The direct and backscattered
primary terms both contain the unknown primary symbol $s_p$ and are therefore
treated as interference. The backscattered secondary component carrying $c_k$
is desired only at the user decoding the \gls{bd} symbol and acts as
interference elsewhere. Backscatter from \gls{bd}$_j$, $j\neq k$, is assumed
negligible owing to the doubly faded path and is absorbed into the effective
noise floor. The received
signal at the \gls{pr} is
\begin{equation}
y_p=\sqrt{P^{\mathrm{p}}}h^{pp}s_p
+\mathbf{h}^{sp,H}\sum_{k=1}^{K}\mathbf{w}_k
\sum_{i=1}^{K_c}\sqrt{\alpha_{k,i}}\,s_{k,i}+n_p,
\label{eq:rx_pri}
\end{equation}
with $n_p\sim\mathcal{CN}(0,\sigma_p^2)$.

Let $C^{p}_{k,i}$ denote the primary interference collected at user
$(k,i)$ and $\Omega_{k,i}$ the residual interference-plus-noise floor it
sees in every decoding event,
\begin{align}
C^{p}_{k,i}&\triangleq
P^{\mathrm{p}}\big|h^{ps}_{k,i}\big|^{2}
+\eta_k P^{\mathrm{p}}\big|f_{k,i}\big|^{2}\big|h^{pb}_k\big|^{2},
\label{eq:cp}\\
\Omega_{k,i}&\triangleq
\sum_{j\neq k}\big|\mathbf{h}_{k,i}^{H}\mathbf{w}_j\big|^{2}
+\eta_k\sum_{j=1}^{K}\big|\mathbf{h}_{k,i}^{c,H}\mathbf{w}_j\big|^{2}
+C^{p}_{k,i}+\sigma^{2},
\label{eq:floor}\\
\Omega^{c}_{k}&\triangleq
\sum_{j\neq k}\big|\mathbf{h}_{k,i^{\star}_k}^{H}\mathbf{w}_j\big|^{2}
+\eta_k\!\sum_{j\neq k}\!\big|\mathbf{h}_{k,i^{\star}_k}^{c,H}\mathbf{w}_j\big|^{2}
+C^{p}_{k,i^{\star}_k}+\sigma^{2},
\label{eq:floor_bd}
\end{align}
where $i^{\star}_k$ denotes the user of cluster $k$ that decodes the \gls{bd}
symbol, identified below once the decoding order has been introduced. Both
terms of $C^{p}_{k,i}$ are independent of $\{\mathbf{w}_k\}$ and
$\boldsymbol{\Theta}$ but scale with $\eta_k$, so a larger reflection
coefficient raises the primary interference seen by every user of cluster $k$.
User $(k,i)$ decodes the symbols of cluster $k$ in the order
$\pi_k$. When it decodes the symbol of user $(k,m)$ with
$\pi_k(m)\le\pi_k(i)$, only the symbols with $\pi_k(m')>\pi_k(m)$ remain, so
the \gls{sinr} at user $(k,i)$ when decoding the symbol of user $(k,m)$ is
\begin{equation}
\gamma_{k,i\to m}=
\frac{\alpha_{k,m}\big|\mathbf{h}_{k,i}^{H}\mathbf{w}_k\big|^{2}}
{\big|\mathbf{h}_{k,i}^{H}\mathbf{w}_k\big|^{2}\,\bar\alpha_{k,m}
+\Omega_{k,i}},
\label{eq:sinr_sic}
\end{equation}
where the arrow notation $i\to m$ reads as user $(k,i)$ decoding the symbol
intended for user $(k,m)$, and
$\bar\alpha_{k,m}\triangleq\sum_{m':\pi_k(m')>\pi_k(m)}\alpha_{k,m'}$ denotes
the residual intra-cluster power, that is, the power carried by the symbols
that remain undecoded when the symbol of user $(k,m)$ is detected. Setting
$m=i$ gives the decoding \gls{sinr} of user
$(k,i)$ for its own symbol, $\gamma_{k,i}\triangleq\gamma_{k,i\to i}$, and the
corresponding rate $R_{k,i}=\log_2(1+\gamma_{k,i})$. The user in the last
decoding position, $i^{\star}_k=\pi_k^{-1}(K_c)$, has cancelled every
\gls{noma} symbol of its cluster and reuses them as a data-aided reference to
detect $c_k$, so that the own-cluster component of the backscattered field is
known while the remaining clusters' symbols stay undecoded. The \gls{sinr} at
$i^{\star}_k$ for the backscatter symbol $c_k$ is
\begin{equation}
\gamma_{k,c}=
\frac{\eta_k\big|\mathbf{h}_{k,i^{\star}_k}^{c,H}\mathbf{w}_k\big|^{2}}
{\Omega^{c}_{k}},
\label{eq:sinr_bd}
\end{equation}
so that the achievable rate of \gls{bd}$_k$ is
$R_{k,c}=\log_2(1+\gamma_{k,c})$.\footnote{The backscattered primary term of
\eqref{eq:rx_sec} carries $c_k$ but reaches the decoding user over a channel
that is not known to it, and is therefore retained as interference within
$C^{p}_{k,i^{\star}_k}$ rather than combined coherently with the wanted
component. Equation~\eqref{eq:sinr_bd} is consequently a conservative estimate
of the achievable backscatter rate.} Let $\nu_{k,i}\geq 0$ and
$\nu_{k,c}\geq 0$ denote the weighting coefficients assigned to NOMA user
$(k,i)$ and \gls{bd}$_k$, respectively. The weighted sum rate of the secondary
network is
\begin{equation}
R^{w}_{\mathrm{sum}}
=\sum_{k=1}^{K}\left(
\sum_{i=1}^{K_c}\nu_{k,i}R_{k,i}
+\nu_{k,c}R_{k,c}
\right).
\label{eq:wsr}
\end{equation}

We hold $\boldsymbol{\alpha}$ and $\boldsymbol{\eta}$ fixed and jointly
optimise the active and passive beamformers; the effect of
$\boldsymbol{\eta}$ is characterised by direct sweep in
Section~\ref{subsec:eta}. The joint optimisation problem is formulated as
\begin{subequations}\label{eq:P1}
\begin{align}
(\mathrm{P1})\ \max_{\{\mathbf{w}_k\},\boldsymbol{\Theta}}\ \
& \sum_{k=1}^{K}\Big(\sum_{i=1}^{K_c}\nu_{k,i}R_{k,i}+\nu_{k,c}R_{k,c}\Big)
\label{eq:P1obj}\\
\text{s.t.}\ \
& \sum_{k=1}^{K}\big|\mathbf{h}^{sp,H}\mathbf{w}_k\big|^{2}\le I_{\mathrm{th}},
\label{eq:P1ith}\\
& \gamma_{k,i\to m}\ge\gamma_{k,m},\ \ \forall k,\ (i,m)\in\mathcal{S}_k,
\label{eq:P1sic}\\
& R_{k,i}\ge R^{\min}_{k,i},\quad\forall k,i,
\label{eq:P1qos}\\
& R_{k,c}\ge R^{\min}_{k,c},\quad\forall k,
\label{eq:P1bdqos}\\
& \sum_{k=1}^{K}\|\mathbf{w}_k\|^{2}\le P_{\max},
\label{eq:P1pow}\\
& \big|[\boldsymbol{\Theta}]_{n,n}\big|=1,\ \ n=1,\dots,N,
\label{eq:P1theta}
\end{align}
\end{subequations}
where $\mathcal{S}_k\triangleq\big\{(i,m):\pi_k(m)<\pi_k(i)\big\}$ denotes the
set of intra-cluster user pairs for which user $(k,i)$ is required to decode
the symbol intended for user $(k,m)$. Constraint~\eqref{eq:P1ith} imposes the
interference-temperature limit at the primary receiver, while
\eqref{eq:P1sic} ensures successful \gls{sic}.
Constraints~\eqref{eq:P1qos}--\eqref{eq:P1bdqos} specify the minimum \gls{qos}
requirements of the \gls{noma} users and backscatter streams, respectively,
whereas \eqref{eq:P1pow} and \eqref{eq:P1theta} enforce the secondary
transmit-power budget and the unit-modulus \gls{ris} constraint. Problem
$(\mathrm{P1})$ is non-convex because the achievable rates are non-concave in
$\{\mathbf{w}_k\}$ due to the interference terms, the active beamforming
variables $\{\mathbf{w}_k\}$ and the \gls{ris} reflection matrix
$\boldsymbol{\Theta}$ are coupled through the cascaded channels in
\eqref{eq:cascaded}, and the unit-modulus constraint in \eqref{eq:P1theta} is
itself non-convex.

\section{Solution of the Problem}
\label{sec:solution}

The decoding order within each cluster is settled first, since it fixes which
symbols appear in the interference terms of every rate expression. With the
order in hand, $(\mathrm{P1})$ separates into an active and a passive
beamforming subproblem, solved in turn; the order is checked again once both
have been updated.

\subsection{Interference-Normalised Channel Gain Based Decoding Order}
\label{subsec:order}

For cluster $k$ under given $\{\mathbf{w}_k\}$ and $\boldsymbol{\Theta}$, the
decoding order is defined by~\cite{cuiUnsupervisedMachineLearningBased2018}
\begin{equation}
\tilde{\gamma}_{k,\pi_k^{-1}(1)}\le\tilde{\gamma}_{k,\pi_k^{-1}(2)}
\le\dots\le\tilde{\gamma}_{k,\pi_k^{-1}(K_c)},
\label{eq:order_rule}
\end{equation}
where $\tilde{\gamma}_{k,i}$ is the interference-normalised equivalent
channel gain
\begin{equation}
\tilde{\gamma}_{k,i}=
\frac{\big|\mathbf{g}_{k,i}^{H}\boldsymbol{\Theta}\mathbf{H}\mathbf{w}_k\big|^{2}}
{\Omega_{k,i}},
\label{eq:gamma_tilde}
\end{equation}
so that $\pi_k$ follows by sorting $\{\tilde{\gamma}_{k,i}\}_{i=1}^{K_c}$
in ascending order. Note that the interference terms in
\eqref{eq:gamma_tilde} carry no power allocation coefficients: since
$\sum_{m=1}^{K_c}\alpha_{j,m}=1$, the inter-cluster and backscatter
contributions of cluster $j$ collapse to
$\sum_{m}|\mathbf{h}_{k,i}^{H}\mathbf{w}_j|^{2}\alpha_{j,m}
=|\mathbf{h}_{k,i}^{H}\mathbf{w}_j|^{2}$, while the intra-cluster
interference is carried by the common beamformer $\mathbf{w}_k$ and is
shared by every user of the cluster. The power allocation coefficients
therefore have no impact on the decoding order, which can be resolved
before the power is allocated.

\begin{lemma}
\label{lem:sic}
Let $i$ and $m$ belong to cluster $k$ with $\pi_k(i)>\pi_k(m)$. Then
$\gamma_{k,i\to m}\ge\gamma_{k,m}$ holds automatically whenever the order
satisfies \eqref{eq:order_rule}.
\end{lemma}

\emph{Proof:} Dividing the numerator and denominator of \eqref{eq:sinr_sic} by
$|\mathbf{h}^{H}_{k,i}\mathbf{w}_k|^{2}$ gives
$\gamma_{k,i\to m}=\alpha_{k,m}\big/\big(\bar\alpha_{k,m}
+1/\tilde\gamma_{k,i}\big)$, the reduction being exact because the denominator
of $\tilde\gamma_{k,i}$ in \eqref{eq:gamma_tilde} is $\Omega_{k,i}$. For a
given target symbol $(k,m)$ the coefficients $\alpha_{k,m}$ and
$\bar\alpha_{k,m}$ are fixed, so the decoding user enters only through
$\tilde\gamma_{k,i}$ and $\gamma_{k,i\to m}$ increases with it. The claim then
follows from $\tilde\gamma_{k,i}\ge\tilde\gamma_{k,m}$, which
\eqref{eq:order_rule} supplies by construction. An analogous result is
established in~\cite{cuiUnsupervisedMachineLearningBased2018} by
cross-multiplication. $\hfill\blacksquare$

Lemma~\ref{lem:sic} indicates that the \gls{sic} constraints \eqref{eq:P1sic}
can be removed under the order \eqref{eq:order_rule}. Problem $(\mathrm{P1})$
under the given decoding order is thus rewritten as
\begin{subequations}\label{eq:P2}
\begin{align}
(\mathrm{P2})\ \max_{\{\mathbf{w}_k\},\boldsymbol{\Theta}}\ \
& \sum_{k=1}^{K}\Big(\sum_{i=1}^{K_c}\nu_{k,i}R_{k,i}+\nu_{k,c}R_{k,c}\Big)\\
\text{s.t.}\ \ & \eqref{eq:P1ith},\ \eqref{eq:P1qos},\ \eqref{eq:P1bdqos},\
\eqref{eq:P1pow},\ \eqref{eq:P1theta}.
\end{align}
\end{subequations}
Since $\tilde{\gamma}_{k,i}$ depends on $\{\mathbf{w}_k\}$ and
$\boldsymbol{\Theta}$, the order is re-evaluated from
\eqref{eq:order_rule} at the end of each iteration. In the following we solve
$(\mathrm{P2})$ instead of $(\mathrm{P1})$, decomposing it into the
active beamforming subproblem of Section~\ref{subsec:bf} and the passive
beamforming subproblem of Section~\ref{subsec:passive}. The backscatter
symbol $c_k$ of \gls{bd}$_k$ is decoded last, at
$i^{\star}_k=\pi_k^{-1}(K_c)$; its decodability is not governed by
Lemma~\ref{lem:sic} but is enforced by \eqref{eq:P1bdqos}.

\subsection{Intra-Cluster Power Allocation and the Equal-Rate Split}
\label{subsec:pac}

The power allocation coefficients $\boldsymbol{\alpha}$ are held fixed in
$(\mathrm{P2})$, so the split within a cluster is a design input rather than an
optimisation variable. We characterise that input in closed form for the
two-user case by identifying the split at which both users of a cluster achieve
equal rate; Section~\ref{subsec:pac_fairness} validates it numerically.

Consider $K_c=2$, and let $i$ and $m$ denote the users of cluster $k$ with
$\pi_k(m)>\pi_k(i)$, so that $m=i^{\star}_k$ is decoded last. Applying to
\eqref{eq:sinr_sic} the same reduction used in the proof of
Lemma~\ref{lem:sic}, the achievable rates become
\begin{align}
&R_{k,i}=\log_2\!\left(1+
\frac{\alpha_{k,i}\tilde{\gamma}_{k,i}}
     {\alpha_{k,m}\tilde{\gamma}_{k,i}+1}\right),
\label{eq:rate_i}\\
&R_{k,m}=\log_2\!\left(1+\alpha_{k,m}\tilde{\gamma}_{k,m}\right).
\label{eq:rate_m}
\end{align}
Here, user $m$ has cancelled the other intra-cluster symbol and therefore
decodes its own signal without intra-cluster interference, whereas user $i$
treats the power allocated to user $m$ as interference.

\begin{proposition}
\label{prop:crossover}
For a given channel realisation and fixed $\{\mathbf{w}_k\}$ and
$\boldsymbol{\Theta}$, the two users achieve equal rate at a unique power
split. The corresponding power allocation of the last-decoded user is
\begin{equation}
\alpha^{\star}_{k,m}=
\frac{-\big(\tilde{\gamma}_{k,m}+\tilde{\gamma}_{k,i}\big)
+\sqrt{\big(\tilde{\gamma}_{k,m}+\tilde{\gamma}_{k,i}\big)^{2}
+4\tilde{\gamma}_{k,m}\tilde{\gamma}_{k,i}^{2}}}
{2\tilde{\gamma}_{k,m}\tilde{\gamma}_{k,i}},
\label{eq:pac_root}
\end{equation}
where $\alpha^{\star}_{k,m}\in(0,1)$ and
$\alpha^{\star}_{k,i}=1-\alpha^{\star}_{k,m}$.
\end{proposition}

\emph{Proof:} Equating \eqref{eq:rate_i} and \eqref{eq:rate_m} and using
$\alpha_{k,i}=1-\alpha_{k,m}$ gives
\[
\alpha_{k,m}^{2}\tilde{\gamma}_{k,m}\tilde{\gamma}_{k,i}
+\alpha_{k,m}\big(\tilde{\gamma}_{k,m}+\tilde{\gamma}_{k,i}\big)
-\tilde{\gamma}_{k,i}=0.
\]
Denoting its left-hand side by $f(\alpha_{k,m})$, we have
$f(0)=-\tilde{\gamma}_{k,i}<0$,
$f(1)=\tilde{\gamma}_{k,m}\tilde{\gamma}_{k,i}+\tilde{\gamma}_{k,m}>0$, and
$f'(\alpha_{k,m})>0$ for $\alpha_{k,m}\ge0$, so the root in $(0,1)$ is
unique.
$\hfill\blacksquare$

Numerically, \eqref{eq:pac_root} is evaluated with its numerator
rationalised, which avoids cancellation when
$\tilde{\gamma}_{k,m}\tilde{\gamma}_{k,i}^{2}$ is small.
Because $\tilde{\gamma}_{k,i}$ and $\tilde{\gamma}_{k,m}$ vary across channel
realisations, so does the split \eqref{eq:pac_root}. Taking the expectation of
the root itself rather than of its arguments, the average fair allocation of
the first-decoded user is defined as
\begin{equation}
\bar{\alpha}^{\star}_{k,i}
\triangleq
\mathbb{E}\!\left[
1-\alpha^{\star}_{k,m}
\big(\tilde{\gamma}_{k,m},\tilde{\gamma}_{k,i}\big)
\right].
\label{eq:alpha_expect}
\end{equation}

\subsection{Active Beamforming Optimisation}
\label{subsec:bf}

Following the auxiliary-variable and \gls{sca} reformulation commonly used
for \gls{ris}-assisted \gls{noma} beamforming~\cite{Mu_Exploiting},
each rate in \eqref{eq:P1obj} is expressed as a logarithm of one plus a
ratio. We split each \gls{sinr} into its numerator and denominator and
introduce separate positive variables: $\tilde A_{k,i}$ absorbs the
reciprocal desired power of \gls{noma} stream $(k,i)$ and
$\tilde B_{k,i}$ its interference-plus-noise total, with
$\tilde A^{c}_{k}$ and $\tilde B^{c}_{k}$ playing the same roles for the
backscatter stream of cluster $k$,
\begin{align}
\tfrac{1}{\tilde A_{k,i}}
&=\big|\mathbf{h}_{k,i}^{H}\mathbf{w}_k\big|^{2}\alpha_{k,i},
&
\tilde B_{k,i}
&=\big|\mathbf{h}_{k,i}^{H}\mathbf{w}_k\big|^{2}\bar\alpha_{k,i}
+\Omega_{k,i},
\label{eq:slack_noma}\\
\tfrac{1}{\tilde A^{c}_{k}}
&=\eta_k\big|\mathbf{h}^{c,H}_{k,i^{\star}_k}\mathbf{w}_k\big|^{2},
&
\tilde B^{c}_{k}&=\Omega^{c}_{k},
\label{eq:slack_bd}
\end{align}
with $\bar\alpha_{k,i}$ the residual intra-cluster power defined after
\eqref{eq:sinr_sic}. Hence
$\gamma_{k,i}=1/(\tilde A_{k,i}\tilde B_{k,i})$ and
$\gamma_{k,c}=1/(\tilde A^{c}_{k}\tilde B^{c}_{k})$.

Using the standard \gls{sdr} lifting~\cite{luo_sdr_2010},
we further define $\mathbf{W}_k=\mathbf{w}_k\mathbf{w}_k^{H}$,
$\mathbf{H}_{k,i}=\mathbf{h}_{k,i}\mathbf{h}_{k,i}^{H}$,
$\mathbf{H}^{c}_{k,i}=\mathbf{h}^{c}_{k,i}\mathbf{h}^{c,H}_{k,i}$ and
$\mathbf{H}^{sp}=\mathbf{h}^{sp}\mathbf{h}^{sp,H}$, where
$\mathbf{W}_k\succeq\mathbf{0}$ and $\mathrm{rank}(\mathbf{W}_k)=1$, so
that
$|\mathbf{h}_{k,i}^{H}\mathbf{w}_j|^{2}=\Tr(\mathbf{W}_j\mathbf{H}_{k,i})$,
$C^{p}_{k,i}$ is a constant of this subproblem, and the floor
\eqref{eq:floor} reads
$\Omega_{k,i}=\sum_{j\neq k}\Tr(\mathbf{W}_j\mathbf{H}_{k,i})
+\eta_k\sum_{j=1}^{K}\Tr(\mathbf{W}_j\mathbf{H}^{c}_{k,i})
+C^{p}_{k,i}+\sigma^{2}$, and likewise for $\Omega^{c}_{k}$ of
\eqref{eq:floor_bd}. For compactness, the optimisation variables of this
subproblem are collected in
\begin{equation}
\boldsymbol{\Psi}\triangleq\big\{\{\mathbf{W}_k\},\{\tilde{A}_{k,i}\},
\{\tilde{B}_{k,i}\},\{\tilde{A}^{c}_{k}\},\{\tilde{B}^{c}_{k}\},
\{\hat R_{k,i}\},\{\hat R_{k,c}\}\big\},
\label{eq:psi_active}
\end{equation}
where $\{\hat R_{k,i}\}$ and $\{\hat R_{k,c}\}$ are the auxiliary rate
variables introduced below. Relaxing \eqref{eq:slack_noma} and
\eqref{eq:slack_bd} to inequalities, which is tight at the optimum since the
objective is monotone in each slack, the active beamforming problem is
\begin{subequations}\label{eq:P3}
\begin{align}
(\mathrm{P3})\ \max_{\boldsymbol{\Psi}}\ \
& \sum_{k=1}^{K}\Big(\sum_{i=1}^{K_c}\nu_{k,i}\hat R_{k,i}
+\nu_{k,c}\hat R_{k,c}\Big)\\
\text{s.t.}\ \
& \log_2\!\Big(1+\tfrac{1}{\tilde A_{k,i}\tilde B_{k,i}}\Big)\ge\hat R_{k,i},
\label{eq:P3_c1}\\
& \log_2\!\Big(1+\tfrac{1}{\tilde A^{c}_{k}\tilde B^{c}_{k}}\Big)\ge\hat R_{k,c},
\label{eq:P3_c2}\\
& \tfrac{1}{\tilde A_{k,i}}\le\Tr(\mathbf{W}_k\mathbf{H}_{k,i})\,\alpha_{k,i},
\label{eq:P3_A}\\
& \tilde B_{k,i}\ge\Tr(\mathbf{W}_k\mathbf{H}_{k,i})\,\bar\alpha_{k,i}
+\Omega_{k,i},
\label{eq:P3_B}\\
& \tfrac{1}{\tilde A^{c}_{k}}\le
\eta_k\Tr(\mathbf{W}_k\mathbf{H}^{c}_{k,i^{\star}_k}),
\label{eq:P3_Ac}\\
& \tilde B^{c}_{k}\ge\Omega^{c}_{k},
\label{eq:P3_Bc}\\
& \sum_{k=1}^{K}\Tr(\mathbf{W}_k\mathbf{H}^{sp})\le I_{\mathrm{th}},
\label{eq:P3_ith}\\
& \sum_{k=1}^{K}\Tr(\mathbf{W}_k)\le P_{\max},
\label{eq:P3_pow}\\
& \hat R_{k,i}\ge R^{\min}_{k,i},\ \ \hat R_{k,c}\ge R^{\min}_{k,c},
\label{eq:P3_qos}\\
& \mathbf{W}_k\succeq\mathbf{0},\ \ \mathrm{rank}(\mathbf{W}_k)=1.
\label{eq:P3_rank}
\end{align}
\end{subequations}

\begin{figure*}[!t]
\normalsize
\begin{align}
\tilde R_{k,i} &=
\log_2\!\Big(1+\tfrac{1}{\tilde A^{(\tau_1)}_{k,i}\tilde B^{(\tau_1)}_{k,i}}\Big)
-\frac{\big(\tilde A_{k,i}-\tilde A^{(\tau_1)}_{k,i}\big)\log_2 e}
      {\tilde A^{(\tau_1)}_{k,i}\big(1+\tilde A^{(\tau_1)}_{k,i}\tilde B^{(\tau_1)}_{k,i}\big)}
-\frac{\big(\tilde B_{k,i}-\tilde B^{(\tau_1)}_{k,i}\big)\log_2 e}
      {\tilde B^{(\tau_1)}_{k,i}\big(1+\tilde A^{(\tau_1)}_{k,i}\tilde B^{(\tau_1)}_{k,i}\big)},
\label{eq:lb_noma}\\
\tilde R_{k,c} &=
\log_2\!\Big(1+\tfrac{1}{\tilde A^{c(\tau_1)}_{k}\tilde B^{c(\tau_1)}_{k}}\Big)
-\frac{\big(\tilde A^{c}_{k}-\tilde A^{c(\tau_1)}_{k}\big)\log_2 e}
      {\tilde A^{c(\tau_1)}_{k}\big(1+\tilde A^{c(\tau_1)}_{k}\tilde B^{c(\tau_1)}_{k}\big)}
-\frac{\big(\tilde B^{c}_{k}-\tilde B^{c(\tau_1)}_{k}\big)\log_2 e}
      {\tilde B^{c(\tau_1)}_{k}\big(1+\tilde A^{c(\tau_1)}_{k}\tilde B^{c(\tau_1)}_{k}\big)}.
\label{eq:lb_bd}
\end{align}
\hrulefill
\vspace*{4pt}
\end{figure*}

Constraints \eqref{eq:P3_c1} and \eqref{eq:P3_c2} are non-convex:
$\log_2(1+1/(AB))$ is jointly convex for $A,B>0$~\cite{Mu_Exploiting}, so
lower-bounding it yields a reverse-convex feasible set. A convex feasible set
is recovered by replacing each logarithm with its tangent plane at the current
iterate, which for a convex function is a global underestimator and therefore
preserves feasibility. Carrying out the expansion in $\tilde A$ and $\tilde B$
about the $\tau_1$-th iterates $\tilde A^{(\tau_1)}_{k,i}$,
$\tilde B^{(\tau_1)}_{k,i}$, $\tilde A^{c(\tau_1)}_{k}$ and
$\tilde B^{c(\tau_1)}_{k}$ gives the minorants $\tilde R_{k,i}$ and
$\tilde R_{k,c}$ shown at the top of the next page in \eqref{eq:lb_noma} and
\eqref{eq:lb_bd}. Relaxing the
rank-one constraint \eqref{eq:P3_rank}, the problem becomes the convex
\gls{sdp}
\begin{subequations}\label{eq:P4}
\begin{align}
(\mathrm{P4})\ \max_{\boldsymbol{\Psi}}\ \
& \sum_{k=1}^{K}\Big(\sum_{i=1}^{K_c}\nu_{k,i}\hat R_{k,i}
+\nu_{k,c}\hat R_{k,c}\Big)\\
\text{s.t.}\ \
& \tilde R_{k,i}\ge\hat R_{k,i},\ \ \tilde R_{k,c}\ge\hat R_{k,c},
\label{eq:P4_rate}\\
& \eqref{eq:P3_A}\text{--}\eqref{eq:P3_qos},\
\mathbf{W}_k\succeq\mathbf{0}.
\end{align}
\end{subequations}

\begin{theorem}
\label{thm:rank_one}
Problem $(\mathrm{P4})$ admits an optimal rank-one solution, i.e.,
\begin{equation}
\operatorname{rank}\!\left(\mathbf{W}_k^{\star}\right)=1,
\qquad \forall k.
\end{equation}
\end{theorem}

\emph{Proof:} The proof follows similar arguments to those
in~\cite{xu_resource_allocation_irs} and is omitted for brevity.

Problem $(\mathrm{P4})$ is a standard convex \gls{sdp} and is solved
efficiently by CVX~\cite{cvx}, with $\mathbf{w}_k^{\star}$
recovered from the principal eigenvector of $\mathbf{W}_k^{\star}$.
Algorithm~\ref{alg:sca} summarises the proposed \gls{sca}-based algorithm,
which terminates when the objective value of $(\mathrm{P4})$ changes by less
than the stopping tolerance $\zeta_1$ and solves each instance to interior-point
accuracy $\epsilon_1$. Feasible initial slacks are obtained by the
slack-minimisation search of~\cite{Lin_initial} and, in subsequent \gls{ao}
iterations, by warm-starting from the previous iterate, with
\eqref{eq:P3_ith} and \eqref{eq:P3_pow} enforced exactly so that every iterate
respects the transmit and incumbent-protection budgets.

\begin{algorithm}[t]
\caption{\gls{sca}-based algorithm for obtaining $\{\mathbf{W}_k^{\star}\}$}
\label{alg:sca}
\begin{algorithmic}[1]
\State \textbf{Input:} stopping tolerance $\zeta_1$, solver accuracy
$\epsilon_1$.
\State Initialise feasible $\{\tilde A^{(0)}_{k,i}\}$,
$\{\tilde B^{(0)}_{k,i}\}$, $\{\tilde A^{c(0)}_{k}\}$,
$\{\tilde B^{c(0)}_{k}\}$, warm-started from the previous \gls{ao}
iteration where available, and set $\tau_1=0$.
\Repeat
  \State Update the slacks and $\{\mathbf{W}_k^{(\tau_1+1)}\}$ by
  solving $(\mathrm{P4})$;
  \State $\tau_1=\tau_1+1$;
\Until{the objective value of $(\mathrm{P4})$ changes by less than $\zeta_1$.}
\State \textbf{Output:} $\{\mathbf{W}_k^{\star}\}$
\end{algorithmic}
\end{algorithm}

\subsection{Passive Beamforming Optimisation}
\label{subsec:passive}

Let $\mathbf{v}=[e^{-j\theta_1},\dots,e^{-j\theta_N}]^{T}$ and
$\mathbf{V}=\mathbf{v}\mathbf{v}^{H}$, with $\mathbf{V}\succeq\mathbf{0}$,
$[\mathbf{V}]_{n,n}=1$ and $\mathrm{rank}(\mathbf{V})=1$. Denoting
$\mathbf{a}_{k,i,j}=\mathrm{diag}\{\mathbf{g}_{k,i}^{H}\}\mathbf{H}\mathbf{w}_j$
and
$\mathbf{a}^{c}_{k,i,j}=f_{k,i}\mathrm{diag}\{\mathbf{g}_{k,b}^{H}\}
\mathbf{H}\mathbf{w}_j$, we have
$|\mathbf{h}_{k,i}^{H}\mathbf{w}_j|^{2}=\Tr(\mathbf{V}\mathbf{A}_{k,i,j})$
and
$|\mathbf{h}^{c,H}_{k,i}\mathbf{w}_j|^{2}
=\Tr(\mathbf{V}\mathbf{A}^{c}_{k,i,j})$, where
$\mathbf{A}_{k,i,j}=\mathbf{a}_{k,i,j}\mathbf{a}_{k,i,j}^{H}$ and
$\mathbf{A}^{c}_{k,i,j}=\mathbf{a}^{c}_{k,i,j}\mathbf{a}^{c,H}_{k,i,j}$,
so that
$\Omega_{k,i}=\sum_{j\neq k}\Tr(\mathbf{V}\mathbf{A}_{k,i,j})
+\eta_k\sum_{j=1}^{K}\Tr(\mathbf{V}\mathbf{A}^{c}_{k,i,j})
+C^{p}_{k,i}+\sigma^{2}$, and likewise for
$\Omega^{c}_{k}$.\footnote{Since the blocked direct links leave the two pure
cascades of \eqref{eq:cascaded} as the only
$\boldsymbol{\Theta}$-dependent quantities, the standard $N$-dimensional
lifting suffices and no augmented coordinate is required. The
interference-temperature limit is likewise absent from this subproblem: with
$\{\mathbf{w}_k\}$ fixed it does not depend on $\boldsymbol{\Theta}$ and is
already satisfied by the active beamforming step.}

Reusing the slacks $\{\tilde A_{k,i},\tilde B_{k,i},\tilde A^{c}_{k},
\tilde B^{c}_{k}\}$ of \eqref{eq:slack_noma} and \eqref{eq:slack_bd}, and
replacing the rate constraints by the lower bounds \eqref{eq:lb_noma} and
\eqref{eq:lb_bd} evaluated at the $\tau_2$-th iterate, the variables of the
passive subproblem are collected in $\boldsymbol{\Phi}$, which is
$\boldsymbol{\Psi}$ of \eqref{eq:psi_active} with the beamforming matrices
$\{\mathbf{W}_k\}$ replaced by the lifted reflection matrix $\mathbf{V}$. The
passive beamforming problem is
\begin{subequations}\label{eq:P5}
\begin{align}
(\mathrm{P5})\ \max_{\boldsymbol{\Phi}}\ \
& \sum_{k=1}^{K}\Big(\sum_{i=1}^{K_c}\nu_{k,i}\hat R_{k,i}
+\nu_{k,c}\hat R_{k,c}\Big)\\
\text{s.t.}\ \
& \tilde R_{k,i}\ge\hat R_{k,i},\ \ \tilde R_{k,c}\ge\hat R_{k,c},
\label{eq:P5_rate}\\
& \tfrac{1}{\tilde A_{k,i}}\le\Tr(\mathbf{V}\mathbf{A}_{k,i,k})\,\alpha_{k,i},
\label{eq:P5_A}\\
& \tilde B_{k,i}\ge\Tr(\mathbf{V}\mathbf{A}_{k,i,k})\,\bar\alpha_{k,i}
+\Omega_{k,i},
\label{eq:P5_B}\\
& \tfrac{1}{\tilde A^{c}_{k}}\le
\eta_k\Tr(\mathbf{V}\mathbf{A}^{c}_{k,i^{\star}_k,k}),
\label{eq:P5_Ac}\\
& \tilde B^{c}_{k}\ge\Omega^{c}_{k},
\label{eq:P5_Bc}\\
& \hat R_{k,i}\ge R^{\min}_{k,i},\ \ \hat R_{k,c}\ge R^{\min}_{k,c},
\label{eq:P5_qos}\\
& \mathbf{V}\succeq\mathbf{0},\ \ [\mathbf{V}]_{n,n}=1,\ \
\mathrm{rank}(\mathbf{V})=1.
\label{eq:P5_rank}
\end{align}
\end{subequations}
The rank constraint \eqref{eq:P5_rank} is the remaining obstacle. Rather than
discard it, we follow the sequential relaxation of~%
\cite{caoSequentialConstraintRelaxation2017} and require only that a fixed
fraction of the energy of $\mathbf{V}$ lie along the leading eigendirection of
the previous iterate,
\begin{equation}
\mathbf{e}_{\max}\big(\mathbf{V}^{(\tau_2)}\big)^{H}\mathbf{V}\,
\mathbf{e}_{\max}\big(\mathbf{V}^{(\tau_2)}\big)
\ge\kappa^{(\tau_2)}\Tr(\mathbf{V}),
\label{eq:srocr}
\end{equation}
which is linear in $\mathbf{V}$ and hence convex. Here
$\mathbf{e}_{\max}(\cdot)$ and $\lambda_{\max}(\cdot)$ return the eigenvector
and the eigenvalue of the largest mode, and $\kappa^{(\tau_2)}\in[0,1]$
measures how strictly rank one is imposed: at $\kappa^{(\tau_2)}=0$ the
constraint is vacuous and at $\kappa^{(\tau_2)}=1$ it is exact. The parameter
is driven upward from zero by
\begin{equation}
\kappa^{(\tau_2+1)}=\min\left(1,\
\frac{\lambda_{\max}\big(\mathbf{V}^{(\tau_2)}\big)}
{\Tr\big(\mathbf{V}^{(\tau_2)}\big)}+\delta^{(\tau_2)}\right),
\label{eq:kappa_update}
\end{equation}
so that each iterate raises the target by a step $\delta^{(\tau_2)}$ above the
ratio $\lambda_{\max}(\mathbf{V}^{(\tau_2)})/\Tr(\mathbf{V}^{(\tau_2)})$
already achieved, this ratio being unity exactly when $\mathbf{V}^{(\tau_2)}$
is rank one. Substituting \eqref{eq:srocr} for
\eqref{eq:P5_rank} yields the convex \gls{sdp}
\begin{subequations}\label{eq:P6}
\begin{align}
(\mathrm{P6})\ \max_{\boldsymbol{\Phi}}\ \
& \sum_{k=1}^{K}\Big(\sum_{i=1}^{K_c}\nu_{k,i}\hat R_{k,i}
+\nu_{k,c}\hat R_{k,c}\Big)\\
\text{s.t.}\ \
& \eqref{eq:P5_rate}\text{--}\eqref{eq:P5_qos},\ \eqref{eq:srocr},
\label{eq:P6_inherit}\\
& \mathbf{V}\succeq\mathbf{0},\ \ [\mathbf{V}]_{n,n}=1,
\label{eq:P6_psd}
\end{align}
\end{subequations}
which is solved efficiently by CVX~\cite{cvx} to accuracy $\epsilon_2$.
Algorithm~\ref{alg:srocr} details the resulting \gls{srocr} procedure, in
which the relaxation tightens towards the rank-one constraint as
$\kappa^{(\tau_2)}\to1$. The procedure stops once $\kappa^{(\tau_2)}$ is
within the rank tolerance $\varrho$ of unity and the objective value changes
by less than the stopping tolerance $\zeta_2$. Upon termination, $\mathbf{V}^{\star}$
satisfies the rank-one condition within the prescribed tolerance $\varrho$. The reflection vector
$\mathbf{v}^{\star}$ is then recovered from the dominant eigencomponent of
$\mathbf{V}^{\star}$, and the phase shifts are obtained as
$\theta_n=-\arg([\mathbf{v}^{\star}]_n)$.

\begin{algorithm}[t]
\caption{\gls{srocr} algorithm for obtaining $\mathbf{V}^{\star}$}
\label{alg:srocr}
\begin{algorithmic}[1]
\State \textbf{Input:} rank tolerance $\varrho$, stopping tolerance
$\zeta_2$, solver accuracy $\epsilon_2$.
\State Initialise feasible $\mathbf{V}^{(0)}$, warm-started from the previous
\gls{ao} iteration where available, step size $\delta^{(0)}$, relaxation
parameter $\kappa^{(0)}=0$ and index $\tau_2=0$.
\Repeat
  \State Solve $(\mathrm{P6})$ to obtain $\mathbf{V}$;
  \If{$(\mathrm{P6})$ is solvable}
    \State $\mathbf{V}^{(\tau_2+1)}=\mathbf{V}$;\ \
    $\delta^{(\tau_2+1)}=\delta^{(0)}$;
  \Else
    \State $\mathbf{V}^{(\tau_2+1)}=\mathbf{V}^{(\tau_2)}$;\ \
    $\delta^{(\tau_2+1)}=\delta^{(\tau_2)}/2$;
  \EndIf
  \State $\tau_2=\tau_2+1$; update $\kappa^{(\tau_2)}$
  via \eqref{eq:kappa_update};
\Until{$|1-\kappa^{(\tau_2)}|\le\varrho$ and the objective value of
$(\mathrm{P6})$ changes by less than $\zeta_2$.}
\State \textbf{Output:} $\mathbf{V}^{\star}$
\end{algorithmic}
\end{algorithm}

\subsection{Proposed Iterative Algorithm}
\label{subsec:ao}

Based on the above discussions, the details of the proposed algorithm for
solving $(\mathrm{P1})$ are given in Algorithm~\ref{alg:ao}, in which the
decoding order is re-evaluated at the end of every iteration from the
beamformers just obtained and the outer loop terminates once the relative
change in the weighted sum rate falls below the stopping tolerance $\zeta_0$.

\begin{algorithm}[t]
\caption{Proposed \gls{ao} algorithm}
\label{alg:ao}
\begin{algorithmic}[1]
\State Initialise $\{\mathbf{w}_k^{(0)}\}$ satisfying \eqref{eq:P3_ith} and
\eqref{eq:P3_pow}, $\boldsymbol{\Theta}^{(0)}$, and stopping tolerance
$\zeta_0$; set $\tau_0=0$.
\State Calculate $\{\tilde\gamma_{k,i}\}$ via \eqref{eq:gamma_tilde} and
obtain $\{\pi_k^{(0)}\}$ according to \eqref{eq:order_rule};
\State Calculate
$R^{w,(0)}_{\mathrm{sum}}\triangleq R^{w}_{\mathrm{sum}}
(\{\mathbf{w}_k^{(0)}\},\boldsymbol{\Theta}^{(0)},\pi_k^{(0)})$;
\Repeat
  \State Update $\{\mathbf{w}_k^{(\tau_0+1)}\}$ via Algorithm~\ref{alg:sca};
  \State Update $\boldsymbol{\Theta}^{(\tau_0+1)}$ via
  Algorithm~\ref{alg:srocr};
  \State Calculate $\{\tilde\gamma_{k,i}\}$ via \eqref{eq:gamma_tilde}
  and update $\{\pi_k^{(\tau_0+1)}\}$ according to \eqref{eq:order_rule};
  \State Calculate $R^{w,(\tau_0+1)}_{\mathrm{sum}}$; $\tau_0=\tau_0+1$;
\Until{$\big|R^{w,(\tau_0)}_{\mathrm{sum}}-R^{w,(\tau_0-1)}_{\mathrm{sum}}\big|
\big/R^{w,(\tau_0)}_{\mathrm{sum}}<\zeta_0$}
\State \textbf{Output:} $\{\mathbf{w}_k^{\star}\}$,
$\boldsymbol{\Theta}^{\star}$
\end{algorithmic}
\end{algorithm}

The Taylor minorants \eqref{eq:lb_noma} and \eqref{eq:lb_bd} are tight at
their expansion point, so the active beamforming update does not decrease the
objective. The \gls{srocr} and decoding-order updates are not guaranteed to
increase it, so the weighted sum rate need not be monotone in every
realisation; it remains bounded by \eqref{eq:P1pow} and \eqref{eq:P1ith}, and
Algorithm~\ref{alg:ao} terminates through $\zeta_0$. In simulation, the average
objective converges within about $12$ iterations, and the final value matches
the best value reached to within $0.01$~bps/Hz on average.

The complexity of Algorithm~\ref{alg:sca} for solving the active
beamforming problem $(\mathrm{P4})$ is
$\mathcal{O}_1\triangleq\mathcal{O}\big(\tau_1^{\max}
\max\{M,(2(L+K)+2)\}^{4}\sqrt{M}\log_2\frac{1}{\epsilon_1}\big)$, where
$\tau_1^{\max}$ is the number of iterations of Algorithm~\ref{alg:sca}
and $\epsilon_1$ the interior-point solution accuracy. The complexity of
Algorithm~\ref{alg:srocr} for solving the passive beamforming problem
$(\mathrm{P6})$ is $\mathcal{O}_2\triangleq\mathcal{O}\big(\tau_2^{\max}
\max\{N,(2(L+K)+N)\}^{4}\sqrt{N}\log_2\frac{1}{\epsilon_2}\big)$, where
$\tau_2^{\max}$ and $\epsilon_2$ are the corresponding iteration count and
solution accuracy, and the $N$ additional constraints are the unit-modulus
equalities $[\mathbf{V}]_{n,n}=1$. The
decoding-order update requires only the sorting of $K_c$ scalars per
cluster, at a cost of $\mathcal{O}(KK_c\log K_c)$, against the
$\mathcal{O}(K_c!)$ of an exhaustive search, and is negligible beside the
two \glspl{sdp}. As a result, the total complexity of Algorithm~\ref{alg:ao} is
$\mathcal{O}(\tau_0^{\max}(\mathcal{O}_1+\mathcal{O}_2))$, where
$\tau_0^{\max}$ is the number of iterations of Algorithm~\ref{alg:ao}.
\begin{table*}[t]
\caption{Simulation Parameters~\cite{zuo_joint_2023,gao_joint_2023}}
\label{tab:sim_parameters}
\centering
\footnotesize
\renewcommand{\arraystretch}{1.05}
\setlength{\tabcolsep}{4pt}
\begin{tabular}{@{}ll@{\hspace{2em}}ll@{}}
\toprule
\textbf{Parameter} & \textbf{Value} & \textbf{Parameter} & \textbf{Value} \\
\midrule
\multicolumn{2}{@{}l}{\textit{Node Coordinates (m)}}
  & \multicolumn{2}{l}{\textit{Channel Parameters}} \\
Base station (BS)        & $(-15,-10,4)$
  & Path-loss exponent, $\beta$              & $2.2$ \\
RIS                      & $(10,5,4)$
  & Reference path loss, $PL(d_0)$, $d_0=1$\,m & $-30$ dB \\
Primary transmitter (PT) & $(50,-10,4)$
  & Rician factor, $K_R$                     & $3$ dB \\
Primary receiver (PR)    & $(30,5,4)$
  & & \\
Cluster 1 centre         & $(14,12,0)$, $r=3$\,m
  & \multicolumn{2}{l}{\textit{System Parameters}} \\
Cluster 2 centre         & $(6,12,0)$, $r=3$\,m
  & BD reflection coefficient, $\eta$        & $0.7$ ($0.9$ in Fig.~\ref{fig:convergence}) \\
                         &
  & Noise power, $\sigma^2$                  & $-90$ dBm \\
\multicolumn{2}{@{}l}{\textit{Network Configuration}}
  & Primary transmit power, $P^{\mathrm{p}}$ & $50$ dBm \\
Clusters, $K$                     & $2$
  & Interference threshold, $I_{\mathrm{th}}$ & $-90$ dBm \\
Users per cluster, $K_c$          & $2$
  & Weights, $\nu_{k,i}$, $\nu_{k,c}$        & $1$~\cite{sun_weighted_2024} \\
Total users, $L=KK_c$             & $4$
  & & \\
NOMA minimum rate, $R_{k,i}^{\min}$ & $1$ bps/Hz
  & \multicolumn{2}{l}{\textit{Algorithm Parameters}} \\
BD minimum rate, $R_{k,c}^{\min}$ & $0.5$ bps/Hz
  & AO tolerance, $\zeta_0$                  & $10^{-3}$ \\
Power budget, $P_{\max}$          & $50$ dBm
  & Channel realisations per point           & $10^{3}$ \\
\bottomrule
\end{tabular}
\end{table*}

\section{Simulation Results}
\label{sec:results}

The simulation parameters are listed in Table \ref{tab:sim_parameters}. Unless
otherwise stated, the channel parameters follow~\cite{zuo_joint_2023,gao_joint_2023}.
All results are averaged over independent channel realisations,
with the BD reflection coefficient held at its default value
except where it is swept.

\subsection{Convergence}
\label{subsec:convergence}
\begin{figure}[t]
  \centering
  \includegraphics[width=\columnwidth]{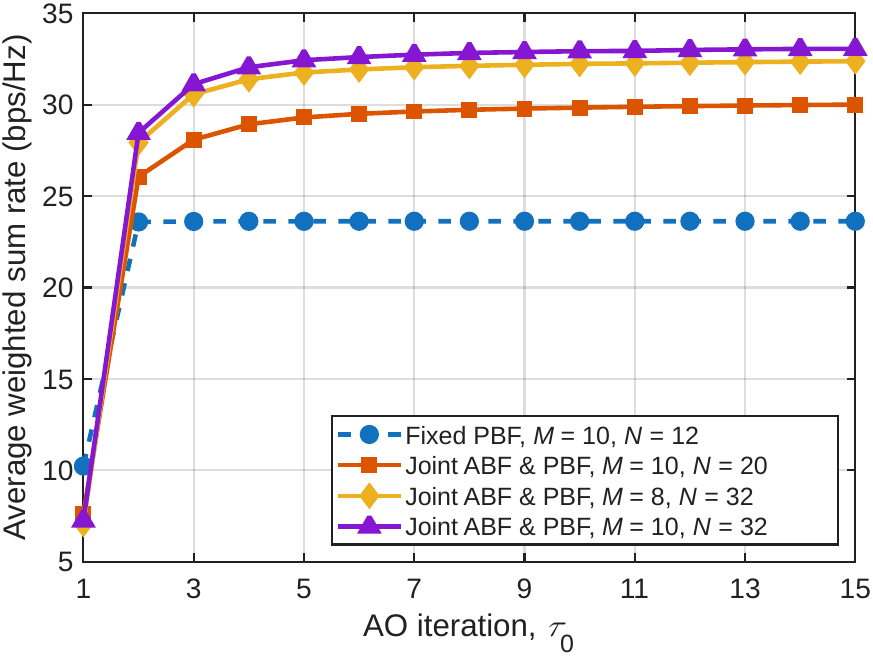}
\caption{Convergence of Algorithm~\ref{alg:ao} for three BS antenna and \gls{ris} element configurations, compared with a fixed passive beamforming (PBF) benchmark in which only the active beamforming (ABF) is optimised.}
  \label{fig:convergence}
\end{figure}
Fig.~\ref{fig:convergence} shows the average weighted sum rate versus the
outer AO iteration index $\tau_0$ for three antenna--RIS configurations,
together with a benchmark in which the RIS phases are kept fixed and only the
active beamformers are optimised. The average objective increases and the relative change in the weighted sum
rate falls below the stopping tolerance $\zeta_0$ within approximately $12$
AO iterations in all cases, confirming the convergence behaviour of
Algorithm~\ref{alg:ao}. The results also show that increasing either the
number of BS antennas or the number of RIS elements improves the achievable
weighted sum rate. In particular, increasing $N$ from $20$ to $32$ provides
a clear gain, while increasing $M$ from $8$ to $10$ for $N=32$ gives a
further improvement, since additional BS antennas provide greater active
beamforming gain and additional RIS elements enhance the passive
beamforming gain. The fixed-PBF benchmark settles after a single iteration
at about $23.6$~bps/Hz, below all jointly optimised configurations, which
indicates the gain from optimising the RIS phases. As the benchmark uses
$N=12$, this comparison is indicative rather than like-for-like.

\subsection{Power Allocation and User Fairness}
\label{subsec:pac_fairness}
We now validate Proposition~\ref{prop:crossover} and quantify the
fairness--throughput trade-off it implies. Throughout this subsection,
$\alpha_{k,i}$ denotes the power allocation of the non-SIC user, while
the complementary allocation of the SIC user is
$\alpha_{k,m}=1-\alpha_{k,i}$.

\begin{figure*}[t]
  \centering
   \includegraphics[width=0.99\textwidth]{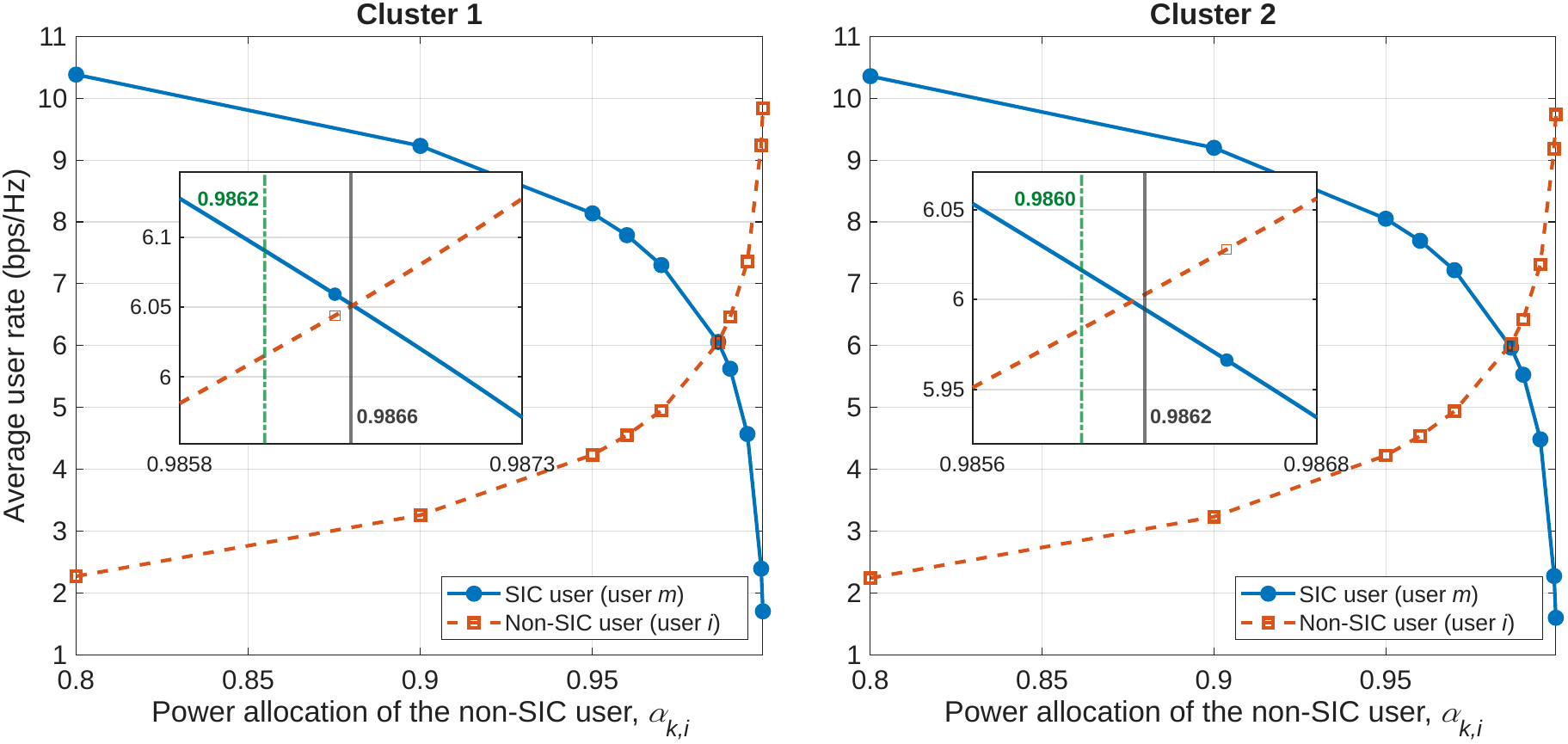}
  \caption{Average user rates against the power allocation of the non-SIC
  user, $\alpha_{k,i}$, with the crossover magnified in the insets. Solid vertical: crossing of the averaged rate
  curves; dash-dot: $\bar{\alpha}^{\star}_{k,i}$ from
  \eqref{eq:alpha_expect}. $M=10$, $N=12$.}
  \label{fig:crossover}
\end{figure*}

Fig.~\ref{fig:crossover} shows the user rates as $\alpha_{k,i}$ is swept over
$[0.8,0.9995]$, with $M=10$, $N=12$ and $1000$ channel realisations per point.
The SIC user has the higher rate over most of the interval, while the
two rates cross close to $\alpha_{k,i}=1$. The measured crossings are $0.9866$
and $0.9862$ for clusters $1$ and $2$, respectively. Evaluating
\eqref{eq:pac_root} in each channel realisation and averaging as in
\eqref{eq:alpha_expect} gives $\bar{\alpha}^{\star}_{k,i}=0.9862$ and $0.9860$,
which agree with the independently observed crossings to within
$4\times10^{-4}$. On average, equal user rates require approximately $98.6\,\%$ of the
cluster power to be allocated to the non-SIC user, leaving only about
$1.4\,\%$ for the SIC user.

Fig.~\ref{fig:equalrate} verifies the condition pointwise by applying
\eqref{eq:pac_root} independently in each channel realisation. The two user
rates are equalised to numerical precision, with mean rates of $6.14$ and
$6.10$~bps/Hz for clusters $1$ and $2$, respectively, and the corresponding
average allocations are consistent with the crossings of
Fig.~\ref{fig:crossover}.
\begin{figure}[t]
  \centering
  \includegraphics[width=\columnwidth]{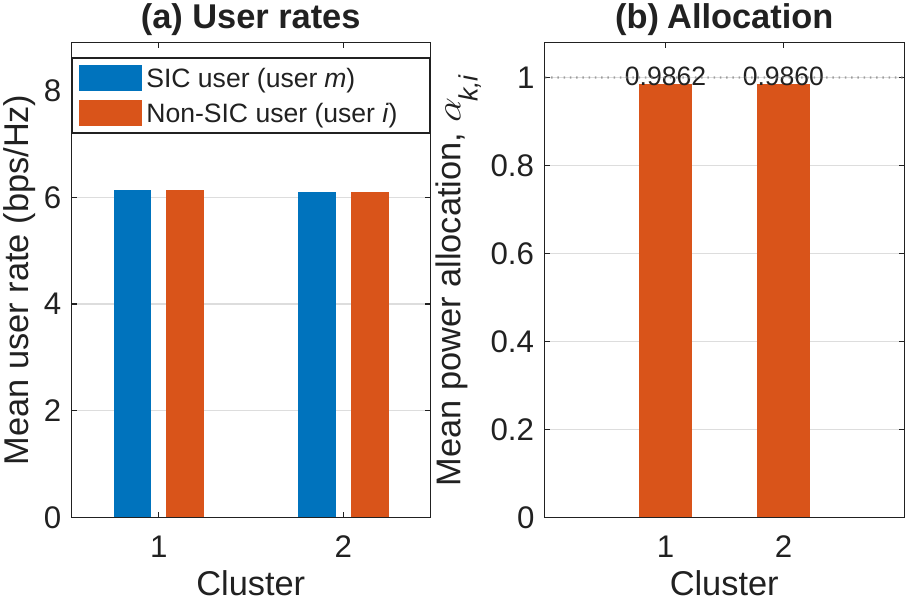}
  \caption{Equal-rate operating point obtained by applying
  \eqref{eq:pac_root} in each channel realisation: (a) mean user rates and
  (b) mean allocation $\bar{\alpha}^{\star}_{k,i}$. $M=10$, $N=12$.}
  \label{fig:equalrate}
\end{figure}

Fig.~\ref{fig:sumrate} illustrates the resulting trade-off. As $\alpha_{k,i}$
increases over the considered range, the total sum rate decreases from $26.57$
to $24.15$~bps/Hz, a $9.1\,\%$ reduction, because power is transferred away
from the SIC user, whose own signal is free of intra-cluster interference
after \gls{sic}. At the equal-rate operating point,
$\alpha_{k,i}\approx0.9865$, the total sum rate is $25.38$~bps/Hz, so rate
equality costs about $4.5\,\%$ relative to $\alpha_{k,i}=0.8$; the full
$9.1\,\%$ applies only when $\alpha_{k,i}$ is pushed beyond that point.
Fig.~\ref{fig:sumrate} also compares a no-\gls{bd} benchmark: the same network
and optimisation with the \glspl{bd} removed, i.e., without backscatter
interference, \gls{bd} rates or \gls{bd} minimum rates. At $\alpha_{k,i}=0.8$, supporting the \glspl{bd}
lowers the \gls{noma} sum rate from $27.65$ to $25.25$~bps/Hz
($8.7\,\%$); their own rate of $1.32$~bps/Hz leaves the total $3.9\,\%$
below the no-\gls{bd} system, which is the cost of the backscatter service.
\begin{figure}[t]
  \centering
  \includegraphics[width=\columnwidth,trim=0 19pt 0 20.5pt,clip]{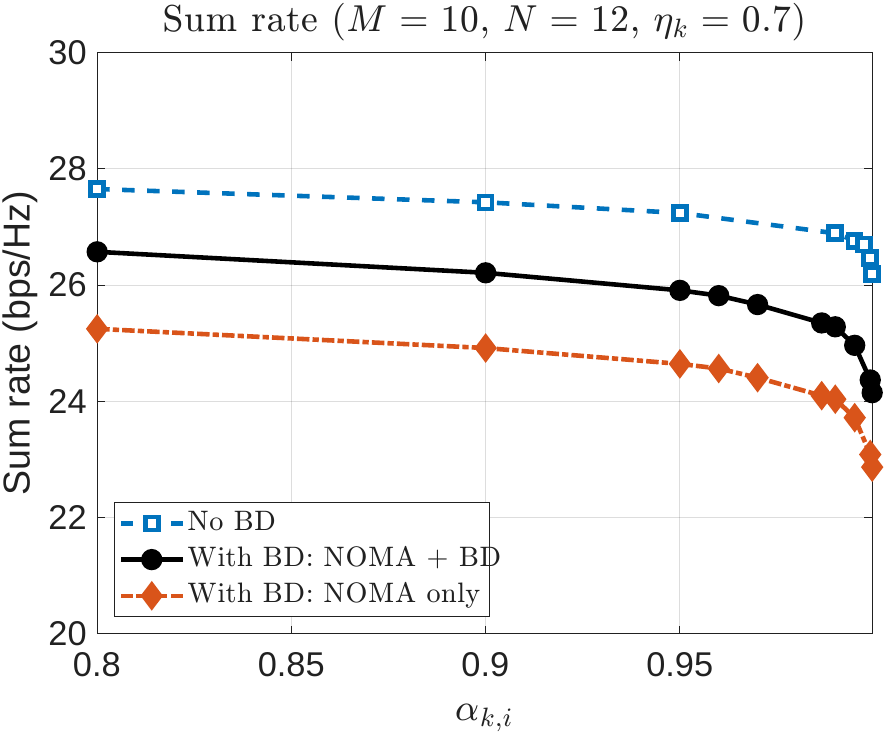}\\[-1pt]
  {\footnotesize Power allocation of the non-SIC user, $\alpha_{k,i}$}
  \caption{Sum rate against the power allocation of the non-SIC user,
  $\alpha_{k,i}$: total and \gls{noma}-only rates with \glspl{bd}, and a no-\gls{bd}
  benchmark (the same network without \glspl{bd}) optimised in the same way.
  $M=10$, $N=12$, $\eta_k=0.7$ for the curves with \glspl{bd}.}
  \label{fig:sumrate}
\end{figure}

\subsection{Effect of the Backscatter Reflection Coefficient}
\label{subsec:eta}

The reflection coefficient $\eta_k$ controls the fraction of incident power
re-radiated by \gls{bd}$_k$ and therefore affects both the desired
backscatter signal and the interference observed by the \gls{noma} users.
To characterise this interaction, $\eta_k$ is swept over
$[0.3,0.999]$ at three intra-cluster power allocations:
$\alpha_{k,m}=0.014$, $0.05$, and $0.20$, with $M=10$ and $N=12$.
The first allocation is close to the equal-rate operating point identified
in Section~\ref{subsec:pac_fairness}, whereas $\alpha_{k,m}=0.20$
represents a more sum-rate-oriented allocation.

\begin{figure}[t]
    \centering
    \includegraphics[width=\columnwidth]{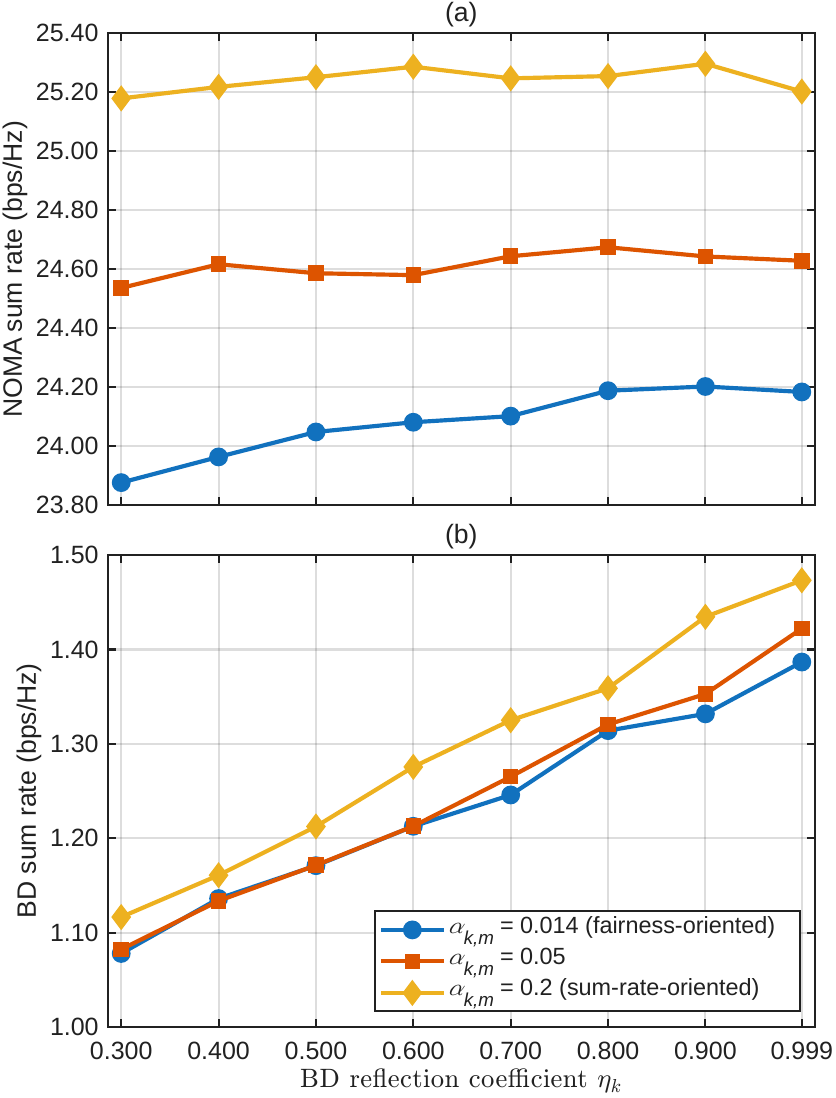}
    \caption{NOMA and BD sum rates against the BD reflection coefficient
    $\eta_k$ for three intra-cluster power allocations:
    (a) NOMA sum rate and (b) BD sum rate. $M=10$, $N=12$.}
    \label{fig:eta}
\end{figure}

As shown in Fig.~\ref{fig:eta}, increasing $\eta_k$ produces a pronounced
improvement in the BD sum rate at all three power allocations, whereas the
aggregate \gls{noma} rate varies only slightly. From $\eta_k=0.3$ to
$0.999$, the BD sum rate increases from $1.078$ to $1.387$~bps/Hz at
$\alpha_{k,m}=0.014$, from $1.082$ to $1.422$~bps/Hz at
$\alpha_{k,m}=0.05$, and from $1.117$ to $1.473$~bps/Hz at
$\alpha_{k,m}=0.20$, corresponding to increases of approximately
$28.6\,\%$, $31.4\,\%$, and $32.0\,\%$, respectively.

In contrast, the corresponding endpoint changes in the \gls{noma} sum rate
are only $1.29\,\%$, $0.38\,\%$, and $0.09\,\%$, respectively. The reflection coefficient therefore strengthens the backscatter
service with only a limited effect on the aggregate \gls{noma} performance,
which suggests a two-stage parameter-selection guideline: $\eta_k$ is first
selected according to the required backscatter-service level. Once $\eta_k$ is fixed, the intra-cluster power allocation can be selected
according to the desired operating objective. For throughput-oriented operation, relatively more power can be assigned
to the stronger SIC user $m$, which decodes its own signal without
intra-cluster interference after \gls{sic}. If user fairness is instead
prioritised, Proposition~\ref{prop:crossover} provides the power split that
equalises the SIC and non-SIC user rates. Thus, $\eta_k$ primarily controls
the backscatter-service level, whereas the power allocation coefficients
determine the operating point between aggregate throughput and
intra-cluster user fairness.

\section{Conclusion}
\label{sec:conclusion}

This paper investigated joint active and passive beamforming for an
RIS-assisted backscatter NOMA network operating as the secondary tier of an
underlay cognitive radio system. An interference-normalised decoding order
was adopted to determine the NOMA decoding sequence, allowing the explicit
SIC constraints to be removed from the optimisation problem. The resulting
non-convex weighted-sum-rate maximisation problem was solved using an
alternating optimisation framework combining SDR-based active beamforming
and SROCR-based passive beamforming. Numerical results showed that the
average weighted sum rate converges within approximately $12$ outer AO
iterations, outperforming a fixed-PBF benchmark, and confirmed that
increasing the numbers of BS antennas and RIS elements improves the achievable
weighted sum rate. For the two-user case, a closed-form power allocation was derived
for equal user rates, and its ensemble-average prediction closely
matched the independently observed rate crossover. The results further revealed a fairness--throughput trade-off: operating at
the equal-rate split reduced the aggregate sum rate by about $4.5\,\%$, and
supporting the \glspl{bd} reduced the \gls{noma} sum rate by $8.7\,\%$ relative
to a no-\gls{bd} design.
The reflection-coefficient study showed that increasing $\eta_k$ improves
the backscatter sum rate substantially while the aggregate \gls{noma} rate
remains comparatively insensitive. Taken together, these results show that
the \gls{bd} reflection coefficient and the intra-cluster power allocation
provide complementary controls for balancing backscatter performance,
aggregate throughput, and user fairness. Future work will compare the
proposed design with benchmark schemes, including \gls{oma}, maximum-ratio and
zero-forcing beamforming, and the no-\gls{cr} special case,
and will examine very small and very large $N$, the Rician factor $K_R$, and
robust designs under partial or statistical \gls{csi}.

\section*{Acknowledgment}
The authors acknowledge the Centre for High Performance Computing (CHPC),
South Africa, for providing computational resources on the Lengau cluster.

\balance

\bibliographystyle{IEEEtran}
\bibliography{references}

\end{document}